\documentclass[10pt,conference]{IEEEtran}
\usepackage{amssymb,amsmath,amsthm, mathtools}
\usepackage{graphicx}
\usepackage{cite}
\usepackage{latexsym}
\usepackage{verbatim}
\usepackage{amsmath}
\usepackage{amssymb}
\usepackage{subfigure}
\usepackage{url}
\usepackage{epstopdf}
\usepackage{algorithm}
\usepackage{times}
\usepackage{stfloats}
\usepackage{multirow}
\usepackage{enumerate}
\usepackage{xspace}
\usepackage{longtable}
\usepackage{mathrsfs}
\usepackage{mathtools}
\usepackage{algcompatible}
\usepackage[noend]{algpseudocode}
\usepackage{amssymb}
\usepackage{pifont}

\usepackage{caption}

\newcommand{\pol}{\ensuremath {\mathit{Pollard}}{\xspace}} 
\newcommand{\gam}{\ensuremath {\mathit{\gamma}}{\xspace}}  
\newcommand{\lam}{\ensuremath {\mathit{\lambda}}{\xspace}} 
\newcommand{\ta}{\ensuremath {\mathit{\tau}}{\xspace}}     
\newcommand{\rsamod}{\ensuremath {\mathit{N}}{\xspace}}    
\newcommand{\nbr}{\ensuremath {\mathit{n}}{\xspace}}       
\newcommand{\srv}{\ensuremath {\mathit{y}}{\xspace}}       

\newcommand{\LPOS}{{\em LPOS}{\xspace}}                    
\newcommand{\PDAFT}{{\em PDAFT}{\xspace}}                  
\newcommand{\ECEG}{{\em ECEG}{\xspace}}
\newcommand{\PPSS}{{\em PPSS}{\xspace}}
\newcommand{\Ra}{\ensuremath {\stackrel{\$}{\leftarrow}{\xspace}}}
\newcommand{\as}{\ensuremath {\leftarrow}{\xspace}}
\newcommand{\sk}{\ensuremath {\mathit{sk}}{\xspace}}
\newcommand{\pk}{\ensuremath {\mathit{PK}}{\xspace}}
\newcommand{\ym}{\ensuremath {\mathit{YM}}{\xspace}}
\newcommand{\ope}{\ensuremath {\mathit{OPE}}{\xspace}}
\newcommand{\yme}{\ensuremath {\mathit{YM.ElGamal}}{\xspace}}

\newcommand{\pr}{\ensuremath {\mathit{\pi}}{\xspace}}
\newcommand{\fc}{\ensuremath {\mathit{FC}}{\xspace}}
\newcommand{\rss}{\ensuremath {\mathit{r}}{\xspace}}

\newcommand{\chn}{\ensuremath {\mathit{chn}}{\xspace}}

\newcommand{\OEnc}[2]{\ensuremath{\mathit{OPE}.\mathit{E}_{#1}\mskip-1mu(#2)}}

\makeatletter
\newcommand{\algrule}[1][.2pt]{\par\vskip.5\baselineskip\hrule height #1\par\vskip.5\baselineskip}
\makeatother

{\bfseries}{\rmfamily}
\newtheorem{definition}{Definition}{\bfseries}{\rmfamily}
\newtheorem{mytheorem}{Theorem}{\bfseries}{\rmfamily}

\IEEEoverridecommandlockouts
\begin{document}

  \title{LPOS: Location Privacy for Optimal Sensing in Cognitive Radio Networks
  \vspace{-0.2in}
}

\author{Mohamed~Grissa$^{\star}$\thanks{This work was supported in part by the US National Science Foundation under NSF CAREER award CNS-0846044.}, Attila Yavuz$^{\star}$, and Bechir Hamdaoui$^{\star}$\\
$^{\star}$\small Oregon State University, grissam,yavuza,hamdaoub@onid.oregonstate.edu\\

\thanks{\copyright~2015 IEEE. Personal use of this material is permitted. Permission from IEEE must be obtained for all other uses, in any current or future media, including reprinting/republishing this material for advertising or promotional purposes, creating new collective works, for resale or redistribution to servers or lists, or reuse of any copyrighted component of this work in other works.}

}

%
%

\maketitle
{\let\thefootnote\relax\footnote{{\\Digital Object Identifier 10.1109/GLOCOM.2015.7417611}}}

\makeatletter

\begin{abstract}
Cognitive Radio Networks (CRNs) enable opportunistic access to the licensed channel resources by allowing unlicensed users to exploit vacant channel opportunities.
%
One effective technique through which unlicensed users, often referred to as Secondary Users (SUs), acquire whether a channel is vacant is cooperative spectrum sensing.
Despite its effectiveness in enabling CRN access, cooperative sensing suffers from location privacy threats, merely because the sensing reports that need to be exchanged among the SUs to perform the sensing task are highly correlated to the SUs' locations. In this paper, we develop a new {\em Location Privacy for Optimal Sensing (\LPOS)} scheme that preserves the location privacy of SUs while achieving optimal sensing performance through voting-based sensing. In addition, \LPOS~is the only alternative among existing CRN location privacy preserving schemes (to the best of our knowledge) that ensures high privacy, achieves fault tolerance, and is robust against the highly dynamic and wireless nature of CRNs. 
\end{abstract}

\section{Introduction}
\label{sec:Introduction}
Cognitive Radio Networks (CRNs) have emerged as a key technology for improving spectrum utilization through opportunistic spectrum access. They do so by allowing unlicensed spectrum users, often referred to as Secondary Users (SUs), to identify and exploit unused opportunities of licensed channels, so long as they do not cause any interference to licensed users, often referred to as Primary Users (PUs)~\cite{hamdaoui2009adaptive}.

Two main approaches can be used by SUs to acquire whether PUs are present in a licensed channel~\cite{akyildiz2011}. The first approach is based on geo-location databases and is very similar to what is used in LBSs (location-based services). The second approach, referred to as cooperative spectrum sensing, relies on the SUs themselves to visit and sense the licensed channels, on a regular basis, to collaboratively decide whether a channel is vacant or not.
In this paper, we focus on the cooperative spectrum sensing approach whose general architecture is shown in Fig.~\ref{fig:initial_network}. In this architecture, the Fusion Center (\fc) is the entity responsible for orchestrating the SUs to perform the sensing task so as to collectively decide whether PUs are present or not.
Through a control channel, \fc~queries SUs, each having sensing capability, to tune to specific channels/frequencies, measure the energy level (known as Received Signal Strength (RSS)) observed in each of these channels, and report the observed RSS values back to \fc\footnote{\scriptsize Energy detection is the most popular method for signal detection due to its simplicity and small sensing time~\cite{fatemieh2011using}.}.
\fc~then first combines the RSS values collected from the different SUs and then compares the combined value against a detection threshold, $\ta$, to decide whether a channel is available. Channel availability decisions are sent back to the SUs to rely on during their opportunistic spectrum access.

\begin{figure}
\center
    \includegraphics[width=0.36\textwidth]{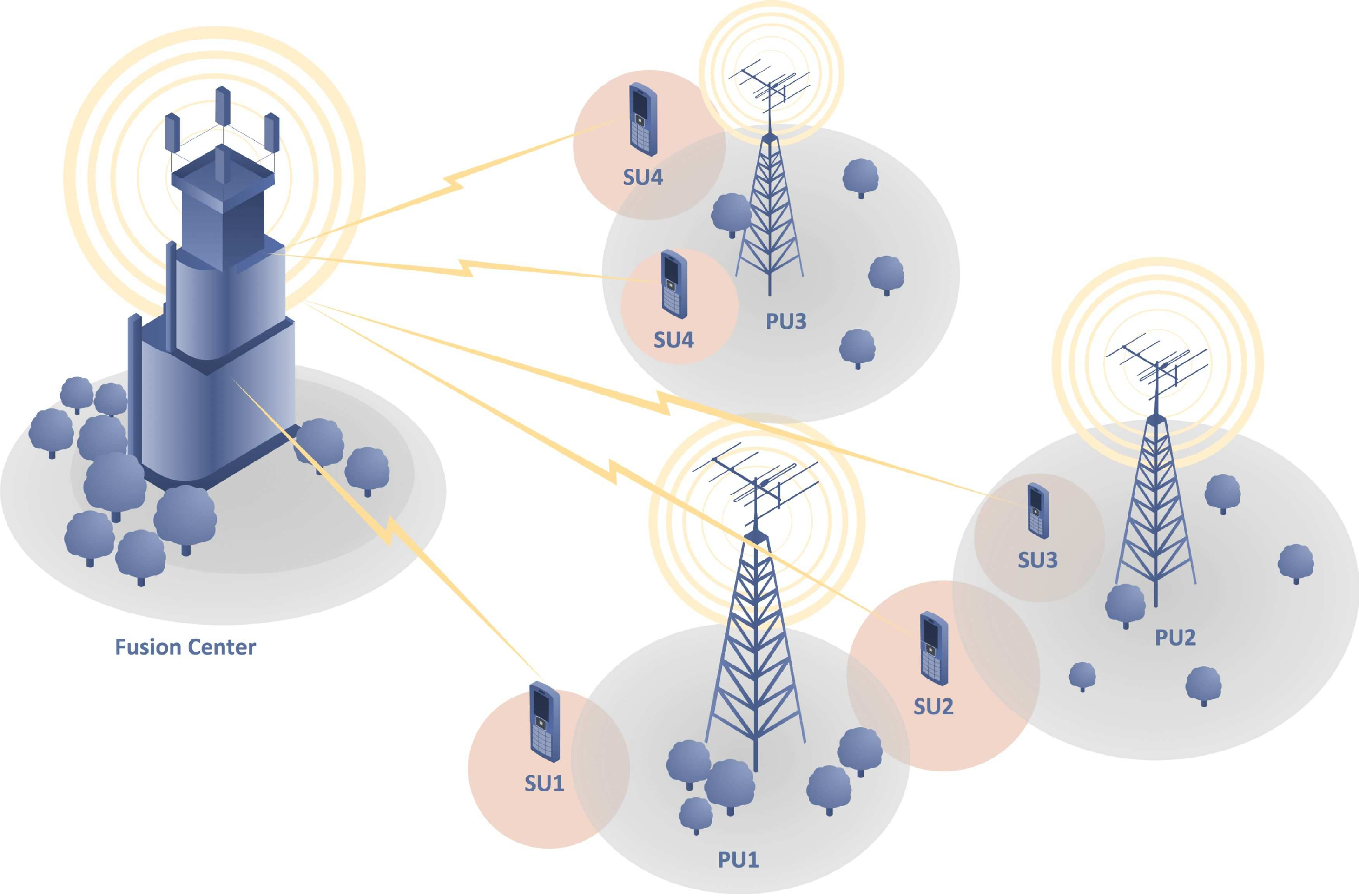}
    \caption{Cooperative spectrum sensing architecture}
    \label{fig:initial_network}
    \vspace{-10pt}

\end{figure}

Despite its effectiveness in improving sensing performance, cooperative sensing suffers from many security and privacy threats that make SUs shy away from participating in the cooperative sensing task. One of these threats is location disclosure. Cooperative spectrum sensing exploits spatial diversity for enhancing accuracy of sensing and this can jeopardize the location privacy of SUs. It has been shown in~\cite{li2012location} that RSS values are heavily correlated to the SUs' physical locations, thus making it not too difficult to compromise the location privacy of SUs.
%
%
Disclosing the location information is undesirable especially when \fc~is run by an untrusted service provider~\cite{bhattacharjee2013vulnerabilities}. The fine-grained location data can be used to determine a lot of information about an individual's beliefs, preferences, and behavior~\cite{wang2014location}. In fact, by analyzing location traces of a user, an adversary can learn that he/she regularly goes to a hospital, and may then sell this information to pharmaceutical advertisers without the user's consent. In addition, malicious adversaries with criminal intent could use this information to pose a threat to an individual's security and privacy. Being aware of such potential privacy risks, SUs may not want to share their data with \fc~s or databases~\cite{wang2014location}, making the need for preserving the location privacy of these users of a high importance.

This paper addresses the SUs' location disclosure threat, considered as one of the most important threats to CRN users' privacy, by designing a protocol that guarantees a high location privacy by concealing the RSS values from \fc~while enabling optimal sensing using the half-voting rule proposed in~\cite{zhang2008cooperative}.

\subsection{Related Work} \label{subsec:RelatedWork}
Although lots of research efforts have already been made when it comes to addressing issues related protocol design~\cite{hamdaoui2008mac}, resource optimization~\cite{venkatraman2010opportunistic,elmachkour2014green,noroozoliaee2013efficient}, spectrum sensing~\cite{akyildiz2011}, and performance modeling and analysis~\cite{elmachkour2015data,adem2014delay}, very little has been made in regards to location privacy issues~\cite{li2012location,chen2014pdaft}.
For instance, Shuai Li et al.~\cite{li2012location} showed that location information of SUs  could be inferred from the sensing reports, and called this attack Single CR Report Location Privacy (SRLP) attack. Another attack in the same context occurs when a user joins or leaves the network. Any malicious entity can estimate the report of a user and hence its location from the variations in the final aggregated RSS measurements when the node joins and leaves the network. This is termed Differential Location Privacy attack.
%
%
To cope with these attacks, the authors propose~\PPSS, a Privacy Preserving collaborative Spectrum Sensing protocol, that uses secret sharing and the Privacy Preserving Aggregation (PPA) process to hide the content of specific sensing reports. It also uses dummy report injections to cope with the Differential Location Privacy attack. However, \PPSS~has several limitations. First, it requires all the sensing reports in order to decode the aggregated result, which makes it quite impractical since the wireless channel may be unreliable, making some sensing reports not accessible by \fc~. Hence, \fc~will not be able to decrypt the aggregated sensing result. Moreover, it cannot cope with the dynamics resulting when multiple users join or leave the network simultaneously. In addition, the pairwise secret sharing process incurs extra communication overhead, which results in an additional delay especially when all the keys need to be updated when a user joins or leaves the network. Also, the encryption scheme used here is practical only when the plaintext space is small, since the decryption of the aggregated result requires solving the DLP problem, which is very costly as shown in Table~\ref{tab:Table3}.



Despite the importance of this issue and the potential that CRNs present, little attention has been paid to this problem. This drove us to look outside the context of CRNs and try to find an approach that might be applied to our setup. We were particularly interested in the work proposed by Chen et al.~\cite{chen2014pdaft} where they present a privacy-preserving data aggregation scheme with fault tolerance for smart grid communications, termed \PDAFT. They considered a setting very similar to the one we study in this work, and tried to preserve users' privacy when smart meters installed within each house sense the consumption information and send it to the control center.

\PDAFT~combines Paillier cryptosystem with Shamir's secret sharing, where a set of smart meters sense the consumption of different households, encrypt their reports using Paillier, then send them to a gateway. The gateway multiplies these reports and forwards the result to the control center, which selects a number of servers (among all servers) to cooperate in order to decrypt the aggregated result. However, \PDAFT~requires a dedicated gateway to collect the encrypted data and a minimum number of working servers in the control center to be able to decrypt the aggregated result. In addition, \PDAFT, like most of the aggregation-based methods, is prone to differential attacks that we mentioned earlier, and does not provide a mechanism that prevents this attack. Another drawback, which is common to simple aggregation-based methods, is that they usually do not provide optimal sensing performance and might be affected by the distribution of the RSS values. Throughout this paper, by optimal sensing we mean final decision accuracy regrading the channel availability.

%

\subsection{Our Contribution} \label{subsec:OurContribution}


\begin{table*}[t!]
\small
\centering  \caption{Privacy, dynamism handling, fault tolerance and sensing performance of our scheme and previous schemes} \label{tab:Table1}
\renewcommand{\arraystretch}{1.2}{

\begin{tabular}{||c|c||c||c||c||c||}

\hline \multicolumn{2}{||c||}{\em Evaluation} & \multicolumn{1}{|c||}{\textbf{Location Privacy}} & \multicolumn{1}{|c||}{\textbf{Dynamism}} & \multicolumn{1}{|c||}{\textbf{Fault Tolerance}} & \multicolumn{1}{|c||}{\textbf{Sensing Performance}}    \\ \hline
\hline  \multicolumn{2}{||c||}{\textbf {{\em Our Scheme:} \LPOS}}  & High & Multiple & yes & optimal \cite{zhang2008cooperative} \\ \hline
\hline  \multirow{2}{*}{\em Generic} & ECC El Gamal & Low & Multiple & yes & not optimal \\ \cline{2-6}
                                     & \PDAFT \cite{chen2014pdaft} & Low & Multiple & yes & not optimal \\ \cline{1-6} \hline \hline

\multicolumn{2}{||c||}{\PPSS~\cite{li2012location}}  & Medium & Single & No & not optimal \\
\hline

\end{tabular}}
\flushleft{\scriptsize{ \textbf{Privacy:} If {\em FC} can learn the aggregated result we evaluate the privacy to be low since an estimation of sensing reports of some users is possible when there are users leaving/joining the network. {\em Medium} privacy if there is a mechanism to cope with the mentioned problem but still using aggregation. We qualify our scheme to have {\em High} privacy since it does not have this vulnerability. { \textbf{Dynamism:}} {\em Multiple} when the scheme can handle multiple users leaving/joining the network simultaneously and {\em Single} when only one SU joining/leaving the network is supported. { \textbf{Fault Tolerance:}} whether or not the system still works normally when one of the SUs fails to send its report. { \textbf{Sensing Performance:} a scheme is optimal if its sensing performance is proven to be optimal otherwise it is not optimal}
}}

\vspace{-3mm}
\end{table*}

We developed \LPOS, a new location privacy for optimal sensing scheme in CRNs. Its main idea lies in enabling privacy-preserving comparison of RSS values and \fc's threshold in an efficient manner via a novel integration of Order Preserving Encryption (OPE)~\cite{boldyreva2009order} and Yao's Millionaires' protocol~\cite{yao1982protocols}. We summarize the key features of our scheme below, compare it in Table~\ref{tab:Table1} to other schemes, and give detailed performance analysis and comparison in Section \ref{sec:PerformanceAnalysis}. 


The key features of \LPOS~are:

\underline{{\em 1) Optimal Sensing}}: To the best of our knowledge, \LPOS~is the first scheme that enables location privacy in CRNs with an optimal spectrum sensing performance. It does so by privacy-preserving realization of the {\em half-voting} rule proposed in \cite{zhang2008cooperative}, which has been shown to be the optimal decision rule for spectrum sensing using energy detection. Unlike aggregation methods that may be vastly impacted by the distribution of RSS values (and misleading \fc~to make inaccurate decisions), this rule enjoys an optimal sensing performance.


\underline{{\em 2) High Location Privacy}}: Unlike some aggregation type protocols~\cite{li2012location,chen2014pdaft}, \LPOS~does not leak RSS information when users join/leave the network, nor does it require dummy report injection to prevent differential attacks as done in~\cite{li2012location}.




\underline{{\em 3) Fault Tolerance}}: In our scheme, if some users cannot sense the channels or fail to send their reports, \fc~only needs to update the voting threshold, $\lam$, with the available users to make an accurate decision. However, some existing schemes cannot handle such failures. For example, \PPSS~\cite{li2012location} requires inputs from all (pre-determined) users to be able to decrypt the aggregated RSS and make a decision. 
\LPOS~does not have such a limitation, since it relies on a voting-based approach and \fc~evaluates each contribution of users towards the decision individually, which makes \LPOS~more fault-tolerant compared to~\PPSS~\cite{li2012location}.

\underline{{\em 4) Scalability and Computational Efficiency}}: \LPOS~offers the smallest communication overhead among its counterparts for large network sizes, and its computational complexity is logarithmic in the number of users, which makes it more practical and scalable (a detailed analysis is given in Section \ref{sec:PerformanceAnalysis}).

\underline{{\em 5) Robust Against Network Dynamism}}: When a group of users join or leave the network, the system security and performance should be maintained. Unlike its counterparts (e.g.,~\PPSS~\cite{li2012location}), which can deal with the joining/leaving of only a single user at a time, \LPOS~can effectively handle multiple, simultaneous join/leave operations.

\section{Preliminaries}
\label{sec:Prelim}


\noindent {\bf CRN System and Sensing Model.}
We consider a centralized CRN that consists of a \fc~and $\nbr$ SUs, as shown in Fig.~\ref{fig:initial_network}.
We assume that each SU is capable of assessing RSS values of channels through energy detection methods~\cite{fatemieh2011using}, and communicating them to \fc, which it then combines them to make decisions regarding whether channels are available. \fc~then broadcasts the final decisions back to SUs.

\noindent {\bf Half-voting rule.}
Two reasons motivated our choice of a voting-based rule over an aggregation-based fusion rule: (i) it has a better sensing performance than aggregation-based rules~\cite{shen2009maximum}, and (ii)
it does not expose users to the privacy issues, we mentioned earlier, that would otherwise be exposed to when aggregation-based rules are used.
The authors in~\cite{zhang2008cooperative} derived a voting threshold, \lam, for optimal spectrum sensing in voting-based CRNs, which is termed half-voting rule. With this, when the number of users whose RSS values are greater than \ta~is higher than $\lam$, then \fc~can conclude that the channel is busy.

\noindent \textbf{Notation.} Operators $||$ and $|x|$ denote the concatenation and the bit length of variable $x$,
respectively.  $x\stackrel{\$}{\leftarrow}\mathcal{S}$ denotes that $x$ is randomly and uniformly selected from the set $\mathcal{S}$. Large primes $q$ and $p>q$ such that $q |
(p-1)$, and a generator $\alpha$ of the subgroup $G$ of order $q$ in $\mathbb{Z}_{p}^{*}$ are selected such that Discrete Logarithm Problem (DLP)~\cite{menezes2010handbook} is intractable. $(\sk,\pk)$~denotes a private/public key pair of ElGamal Encryption~\cite{ElGamal:1985:PKC:19478.19480}, generated under $(G,p,q,\alpha)$. $c\as\OEnc{K}{M}$ denotes order preserving encryption (as defined in Definition~\ref{def:OPE})~of a message $M\in\{0,1\}^{d}$ under private key $K$, where integer $d$ is the block size of \ope.


\noindent \textbf{Cryptographic Building Blocks.} Our scheme utilizes various cryptographic building blocks, which are described below:


$\bullet$~{\em Order Preserving Encryption~(\ope)~\cite{boldyreva2009order}:} 

\begin{definition}~\label{def:OPE}
An \ope~is a deterministic symmetric encryption scheme whose encryption operation preserves the numerical ordering of the plaintexts, i.e. for any two messages $m_1$ and $m_2 \:~s.t.~\: m_1 \leq m_2$, we have $c_1\as\OEnc{K}{m_1}$ $\leq c_2\as\OEnc{K}{m_2}$.
\end{definition}

The \ope~concept was first formalized by Boldyreva et. al~\cite{boldyreva2009order}. Note that our scheme can use {\em {any secure}} \ope~scheme (e.g.,~\cite{kerschbaum2014optimal}) as a building block, and receive the benefits of the security enhancement (e.g.,~\cite{kerschbaum2014optimal}). However, we chose the publicly available implementation of Boldyreva's scheme~\cite{boldyreva2009order} so as to evaluate our scheme in terms of execution time. In \cite{boldyreva2009order}, an ideal security notion, called \emph{indistinguishability under ordered chosen-plaintext attack (IND-OCPA)}, was introduced, which implies that \ope~has no leakage, except the order of ciphertexts. However, Boldyreva et. al~\cite{boldyreva2011order} showed that the ideal \ope~security is unachievable, since it requires a ciphertext size that is at least exponential in the size of the plaintext, leading to the introduction and adoption of a weaker security notion of {\em Random Order-Preserving Functions (ROPF)}, as defined below. 

\begin{definition} \label{def:OPESecurity}
An \ope~based on {\em ROPF}~leaks the order of plaintexts and also at least half of the high-order bits of the plaintext~\cite{boldyreva2011order}.
\end{definition}


$\bullet$~{\em Secure Comparison Protocol:} The Yao's Millionaires' (\ym) protocol~\cite{yao1982protocols} enables two parties to execute ``the greater-than" function, $GT(x, y) = [x > y]$, without disclosing any other information apart from the outcome of the comparison. We use an efficient \ym~scheme~\cite{lin2005efficient}, referred to as \yme, which ensures that only the initiator learns the outcome.

 \begin{definition} \label{def:yme}
Let $(\mathcal{X},\mathcal{Y})$ and $(x,y)\in\{0,1\}^{l}$ be two parties and $l$-bit integers to be compared, respectively. Let $\pr=(l,q,p,\alpha,$ $\{\sk,\pk\})$ be \yme~parameters generated by the protocol initiator $\mathcal{X}$. \yme~returns a bit $b\as\yme$ $(x,y,\pr)$, where $b=0$ if $x<y$ and $b=1$ otherwise. Only $\mathcal{X}$ learns $b$ but $(\mathcal{X},\mathcal{Y})$ learn nothing else. \yme~is secure in the semi-honest setting if ElGamal encryption scheme~\cite{ElGamal:1985:PKC:19478.19480} is secure.
\end{definition}

%

$\bullet$~{\em Group Key Establishment and Management:} We use a dynamic and contributory group key establishment and management protocol for secure group communication purposes.

\begin{definition} \label{def:GroupKey}
{\em Tree-based Group Elliptic Curve Diffie-Hellman} {\em (TG ECDH)}~\cite{wang2006performance} permits $n$ distinct users to collaboratively establish and update a common group key $K$ by extending 2-party ECDH key exchange protocol to $n$-party. {\em TGECDH} is secure if Elliptic Curve Discrete Logarithm Problem (ECDLP) is intractable~\cite{wang2006performance}.
\end{definition}

\section{The Proposed Scheme}
\label{sec:ProposedSchemes}

Voting-based spectrum sensing offers several advantages over its aggregation-based counterparts as discussed in Section \ref{sec:Prelim}. However, this approach requires comparing \fc's threshold $\ta$ and the RSS value $r_i$ of each user $U_i$, thereby forcing at least one of the parties to expose its information to the other. One solution is to use a secure comparison protocol, such as \yme, between \fc~and each user $U_1,\ldots,U_{\nbr}$ in the network, which permits \fc~to learn the total number of users above/below threshold $\ta$ (as discussed in Section \ref{subsec:RelatedWork}) but nothing else. However, secure comparison protocols involve several costly public key crypto operations (e.g., modular exponentiation), and therefore $\mathcal{O}(\nbr)$ invocations of such a protocol per sensing period incur prohibitive computational and communication overhead.

\begin{algorithm}[p]
\caption{\LPOS~Algorithm (the proposed scheme) }\label{alg1}
\begin{algorithmic}[1]

\Statex   \textbf{Initialization}: Executed only once at the beginning.
\State \fc~sets its energy sensing and optimal voting thresholds \ta~and $\lam$, respectively as in~\cite{zhang2008cooperative}. Bit-length $\gam=|\ta|=|\rss_i|$ for $i=1,\ldots,\nbr$, where $\rss_i$ denotes RSS value of user $U_i$.
\State \fc~generates \yme~parameters \pr~(as defined in Definition \ref{def:yme}) and pre-computes ElGamal encryption values in \pr~based on \ta~to accelerate \yme~protocol. \fc~also generates a random padding $D\Ra\{0,1\}^{d-\gam-1}$, where $d$ is the block size of \ope. $D$ is known to all users.
\State There are $\nbr$ users $\{U_i\}_{i=1}^{\nbr}$ in the system, whose RSS values are denoted as $\rss_i$ for $i=1,\ldots,\nbr$, respectively.
\State $\mathcal{G}=\{U_i\}_{i=1}^{\nbr}$ collaboratively establish a group key $K$ via {\em TGECDH} (Definition \ref{def:GroupKey}). 
\State \fc~establishes an authenticated secure channel $\chn_i$~with each user $U_i$ for $i = 1,\ldots,\nbr$. \Comment {(e.g., via SSL/TLS) }.

\hspace{20pt}\algrule
\Statex   \textbf{Private Sensing}: Executed every sensing period $t_w$
\State $U_i$ computes $c_i \as \OEnc{K}{D||\rss_i}$ for $i=1,\ldots,\nbr$.
\State $U_i$ sends $c_i$ to \fc~over $\chn_i$ for $i=1,\ldots,\nbr$.
\State \fc~sorts encrypted RSS values as $c_{min} \leq \ldots \leq c_{max}$ (by Definition \ref{def:OPE}). \label{alg1:line:pt}
\State \fc~initiates \yme~as $b\as\yme$ $(\ta,\rss_{id_{max}},\pr)$ with user $id_{max}$ having the maximum $c_{max}$.
\If {$b=1$}
\State $decision$ $ \gets $ Channel is free.
\Else
\State \fc~initiates \yme~as $b\as\yme(\ta,$ $\rss_{id_{min}},\pr)$ with user $id_{min}$ having the minimum $c_{min}$.
\If {$b=0$}
\State $decision$ $ \gets $ Channel is busy.
\Else
\State \fc~initiates \yme~with a subset of users based on a binary search of \ta~on the remaining encrypted RSS values as described below. Let index $I$ be the index of user $c_{I}$, where \yme~with binary search process is finalized (i.e., $\rss_{I-1} \leq \ta \leq \rss_I$). 
\State \fc~counts the number of $U_i$s s.t. $\ta \leq \rss_i$ : $z \gets n-I$
\If {$z \geq  \lam$} 
\State $decision$ $ \gets $ Channel is busy
\Else
\State $decision$ $ \gets $ Channel is free
\EndIf
\EndIf
\EndIf
\Return $decision$

\hspace{20pt}\algrule
\Statex   \textbf{Group Membership Change Update}:
\State If new user(s) join/leave $\mathcal{G}$ in $t_w$, the new set of users $\mathcal{G}'$ forms a new group key $K'$ via {\em TGECDH}. \fc~may update its threshold and \yme~parameters as \lam'~and \pr' when required.
\State Follow the private sensing steps with new $(K',\lam',\pr')$.
\end{algorithmic}
\end{algorithm}

The key observation that led us to overcome this challenge is the following: If we enable \fc~to learn the relative order of RSS values but nothing else, then the number of \yme~invocations can be reduced drastically. That is, {\em the knowledge of relative order permits \fc~to execute \yme~protocol at worst-case $\mathcal{O}(log(\nbr))$ by utilizing a binary-search type approach}, as opposed to running \yme~with each user in total $\mathcal{O}(\nbr)$ overhead.

This simple yet powerful observation enables us to develop \LPOS, which achieves the above objective via an innovative integration of \ope~scheme, {\em TGECDH} and \yme~protocols. {\em The crux of the idea is to make users \ope~encrypt their RSS values under a group key $K$, which is derived via {\em TGECDH} at the beginning of the sensing period}. In this way, \fc~can learn the relative order of encrypted RSS values but nothing else (and users do not learn each others' RSS values, as they are sent to \fc~over a pairwise secure channel). \fc~then uses this knowledge to run  {\em \yme~protocol by utilizing a binary-search strategy}, which enables it to identify the total number of users above/below threshold $\ta$ (as defined by voting-based optimal sensing in~\cite{zhang2008cooperative}) with only $\mathcal{O}(log(\nbr))$ complexity. This strategy makes \LPOS~the only alternative among its counterparts that can achieve CRN location privacy with an optimal spectrum sensing, fault-tolerance and network dynamism simultaneously (as discussed in Section \ref{subsec:OurContribution} and Section \ref{sec:PerformanceAnalysis}).

We give the detailed description of \LPOS~in Algorithm \ref{alg1}, and further outline the high-level description of \LPOS~as below:

 $\bullet$~{\em Initialization:} \fc~sets up spectrum sensing and crypto parameters for cryptographic building blocks. Users establish a group key $K$ via {\em TGECDH}, with which they will \ope~encrypt their RSS values during the private sensing. \fc~also establishes a secure channel $\chn_i$ with each user $U_i$.

 $\bullet$~{\em Private Sensing:} Each user $U_i$ \ope~encrypts its RSS value $r_i$ with group key $K$ and sends ciphertext $c_i$ to \fc~over $\chn_i$. This permits \fc~to sort ciphertexts as $c_{min} \leq \ldots \leq c_{max}$ without learning corresponding RSS values, and the secure channel $\chn_i$ protects the communication of $U_i$ from other users (as each $r_i$ is encrypted under the same $K$) as well as from outside attackers. \fc~then initiates \yme~first with the user that has the highest RSS value $\rss_{max}$. If it is smaller than energy sensing threshold $\ta$ then the channel is free. Otherwise, \fc~initiates \yme~with the user that has $\rss_{min}$. If it is bigger than $\ta$ then the channel is busy. Otherwise, to make the final decision based on the optimal sensing threshold $\lam$, \fc~runs \yme~according to the binary-search strategy as described in Steps 17-22, which guarantees the decision at the worst $\mathcal{O}(log(\nbr))$ invocations.

 $\bullet$~{\em Update Private Sensing after Group Membership Changes:} At the beginning of each sensing period $t_w$, according to the membership changes in the user group, a new group key may be formed via the update procedure of {\em TGECDH} efficiently. \fc~may optionally update sensing parameters. The private sensing for the new sensing period then begins with new parameters and group key $K'$ and is executed as described above.

\section{Security Analysis}
\label{sec:SecurityAnalysis}

\textbf{Threat Model:} Our threat model focuses on the {\em location privacy (i.e., RSS values) of SUs}.  We consider {\em honest but curious (semi-honest)} setting for \fc~and SUs forming group $\mathcal{G}$(no party, including \fc, maliciously modifies the integrity of its input). This means that they execute the protocol honestly but will show interest in learning information about the other parties. That is, \fc~and other SUs in the group $\mathcal{G}$ may target the location information of a SU $U_i$. RSS value $r_i$ of $U_i$ reveals this location information and therefore should be protected. SUs also may target the threshold value \ta~of \fc. However, we assume that \fc~does not collude with some SUs to localize the other SUs, nor do SUs collude with each others or expose the group key $K$ to \fc~or external parties maliciously. Similarly, we assume that \fc~and SUs do not inject false \ta~or RSS values into spectrum sensing. Finally, an external attacker $\mathcal{A}$ may launch passive attacks against the output of cryptographic operations and active attacks including packet interception/modification to \fc~and SUs. We rely on traditional authenticated secure channel to prevent such an external attacker $\mathcal{A}$.

\textbf{Security Objectives and Analysis:} 
\begin{definition} \label{def:SecObjectives}
Under our threat model described above, \LPOS~security objectives are: (i) RSS values $r_i$ of each $U_i$ remain confidential during all sensing periods. (ii) The sensing threshold \ta~of \fc~remains confidential for all sensing periods. (iii) A secure channel is maintained between each SU and \fc. (iv) Objectives (i)-(iii) are maintained for every membership changes in $\mathcal{G}$. 
\end{definition}

It is easy to show that \LPOS~is secure according to Definition \ref{def:SecObjectives}, as long as its underlying cryptographic blocks are secure.
\begin{mytheorem}
\LPOS~achieves security objectives in Definition~\ref{def:SecObjectives}, as long as \ope, \yme~and {\em TGECDH} are secure according to Definitions~\ref{def:OPESecurity},~\ref{def:yme} and \ref{def:GroupKey}, respectively.
\end{mytheorem}
\begin{proof}
In sensing period $t_w$, objectives (i)-(iv) in Definition \ref{def:SecObjectives} are achieved as follows:

{\em \underline{Initialization}:} In Step 1-2, \fc~sets up system and security parameters such that \yme~and \ope~are secure. Padding $D$ and proper block size of \ope~ensures the leftmost bit leakage from \ope~as defined in Definition \ref{def:OPESecurity} does not leak RSS value during the private sensing. In Step 4, SUs establish a group $K$, which protects $r_i$ values against \fc~via \ope~encryption (as required by (i) in Definition~\ref{def:SecObjectives}). In Step 5, \fc~and each $U_i$ establish a secure channel, which protects \ope~encrypted $r_i$ values $c_i$ (under the same group key $K$) from other SUs and external attacker $\mathcal{A}$ (as required by (i) and (iii) in Definition~\ref{def:SecObjectives}).

{\em \underline{Private Sensing}:} \ope~encryptions in Step 6 ensure the confidentiality of $r_i$ values against \fc~during the ciphertext sorting ($c_1$,\ldots,$c_{\nbr}$) in Step 8, as long as \ope~is secure according to Definition \ref{def:OPESecurity} (with proper padding and \ope~block size as set in the initialization phase). Step 7 ensures the confidentiality of $r_i$ of $U_i$ against other SUs as well as the protection of the communication against an external attacker $\mathcal{A}$ via the secure channel. Hence, objective (i) in Definition~\ref{def:SecObjectives} is achieved during \ope~phase of \LPOS. Step 9 - Step 22 execute \yme, which leaks no information on $\ta$ to SUs and $r_i$'s to \fc~as required. Hence, objectives (i)-(iii) in Definition \ref{def:SecObjectives} are achieved during the whole private sensing steps.

{\em \underline{Update Private Sensing after Group Membership Changes}:} Step 23 ensures that a new group key $K'$ (based on Definition \ref{def:GroupKey}) and parameters $(\lam',\pr')$ are generated according to the membership status of the new group $\mathcal{G'}$. Step 24 ensures the private sensing steps are executed using new $(K',\lam',\pr')$ for each new sensing period. Thus, security objectives (i)-(iv) in Definition \ref{def:SecObjectives} are achieved for all sensing periods as required.
\end{proof}

\section{Analysis and Comparison}
\label{sec:PerformanceAnalysis}

\begin{table*}[t!]
\small
\centering  \caption{Communication overhead, computation cost and storage needed of our scheme and previous schemes} \label{tab:Table3}

\renewcommand{\arraystretch}{1.22}{
\resizebox{\textwidth}{!}{%
\begin{tabular}{||c|c||c||c|c||c|c||}

\hline \multicolumn{2}{||c||}{\multirow{2}{*}{\em Evaluation}}   &  \multicolumn{1}{|c||}{\multirow{2}{*}{\textbf{Communication}}} & \multicolumn{2}{|c||}{\textbf{Computation}}    & \multicolumn{2}{|c||}{\textbf{Storage}} \\ \cline{4-7}

\multicolumn{2}{||c||}{}  &  &  \textbf{ {\em FC}} & \textbf{ {\em SU}} & \textbf{ {\em FC}} & \textbf{ {\em SU}} \\ \hline
\hline  \multicolumn{2}{||c||}{\textbf{{\em Our Scheme:} \LPOS}} & $ 2\gam \cdot |p| \cdot (2+log\:\nbr)+ \nbr \cdot \epsilon_{\ope}+|Q| \cdot log\: \nbr$  &  $1/2 \cdot(2+log\:\nbr)\cdot\gam \cdot|p|\cdot Mulp$ & $(2\gam\cdot |p|+2\gam)\cdot Mulp+\ope +2\:log\: \nbr \cdot PMulQ$ &$4|p|$ & $|p|+|K|$ \\ \hline
\hline  \multirow{2}{*}{\em Generic} & \ECEG & $4|Q| \cdot \nbr$ &$PMulQ + PAddQ + \sqrt{\nbr\cdot \delta} \cdot \pol$ & $2PMulQ + (\nbr-1)\cdot PAddQ$ &$(1+2 \cdot O(log(n\delta))) \cdot |Q|$ & $|Q|$  \\ \cline{2-7}
                                     & \PDAFT \cite{chen2014pdaft} & $2|\rsamod|\cdot (\nbr+1)$ & $2Exp\rsamod^2 + Inv\rsamod^2 + \srv \cdot Mul\rsamod^2$ &$ 2Exp\rsamod^2 + Mul\rsamod^2$ &$2|\rsamod|$ & $|\rsamod| + |\rsamod^2|$   \\ \cline{1-7}

\hline \hline \multicolumn{2}{||c||}{\PPSS\cite{li2012location}} & $ |p|\cdot \nbr$  &  $H + (\nbr+2) \cdot Mulp + (2^{\gam-1}\cdot \nbr + 2) \cdot Expp$ & $H + 2Expp + Mulp$ &  $(\nbr+1)\cdot |p|$ & $(\nbr+1)\cdot |p|$ \\

\hline

\end{tabular}}}

\flushleft{\scriptsize{ \textbf{(i) Variables:} $\gam$: size of the sensing reports, $\nbr$: number of $SUs$, $\rsamod$: modulus in Paillier, $p$: modulus of El Gamal, $H$: cryptographic hash operation, $K$: group key used in \ope. $Expu$ and $Mulu$ denote a modular exponentiation and a modular multiplication over modulus $u$ respectively, where $u \in \{\rsamod, \rsamod^2, p\}$. $Inv\rsamod^2$: modular inversion over $\rsamod^2$, $PMulQ$: point multiplication of order $Q$, $PAddQ$: point addition of order $Q$. $\srv$: number of servers needed for decryption in \PDAFT.
\textbf{(ii)  Parameter size:} For a security parameter $\kappa = 80$, suggested parameter sizes by {\em NIST 2012} are given by : $|\rsamod| = 1024,\; |p| = 1024, \; |Q|=160$ as indicated in \cite{keylength}.
\textbf{(iii) OPE:} the computational complexity of the \ope~is given by $\ope = (log\:|\mathscr{C}|+1)\cdot T_{HGD} + (log\:|\mathscr{P}|+3)\cdot(5log\:|\mathscr{C}| +\theta'+1)/128\cdot T_{AES} $, where $\mathscr{P}$,$\mathscr{C}$ are plaintext and ciphertext spaces respectively and $\theta'$ is a constant. $\epsilon_{\ope}$ is the maximum ciphertext size that could be obtained under the \ope~encryption. This value was determined experimentally based on the \ope~implementation in \cite{opeRuby} and we noticed that it doesn't exceed the 128bits block size of the underlying AES block cipher $\Rightarrow \epsilon_{\ope} = 128\:bits$. \textbf{(iv) ECEG:} The SUs use the FC's \ECEG~public key to encrypt their RSSs and then one node is picked to collect the ciphertexts and multiply them together including its own encrypted RSS and then send the result to the FC. The decryption of the aggregated message in \ECEG~is done by solving the constrained ECDLP problem on small plaintext space similarly to \cite{li2012location} via Pollard's Lambda algorithm, which requires $O(\sqrt{n \cdot\delta})\cdot \pol$ computation and $O(log(n\delta))$ storage \cite{menezes2010handbook}, where $\delta = a-b$ if $RSS \in [a,b]$ and \pol \;is the number of point operations in Pollard Lambda algorithm which varies depending on algorithm implementation used. \textbf{(v) YM.ElGamal:} The communication cost for one comparison is $4\gam \cdot |p|$. The total computational cost of the scheme for one comparison is $5\gam \:log\:p +2\nbr$. Since in our scenario the value of the energy threshold $\ta$ remains unchanged, we can encrypt it only once and offline so the encryption cost can be omitted and the new total computational cost would be $3\gam \cdot |p|+2\gam$ for each comparison operation. \textbf{(vi) TGECDH}: It permits the alteration of group membership (i.e., join/leave), on average $\mathcal{O}(log(n))$ communication and computation (i.e., ECC scalar multiplication) \cite{steiner1996diffie}.
}}

\vspace{-3mm}

\end{table*}

\begin{figure*}[htp]
  \centering

  \subfigure[Communication overhead]{\includegraphics[scale=0.36]{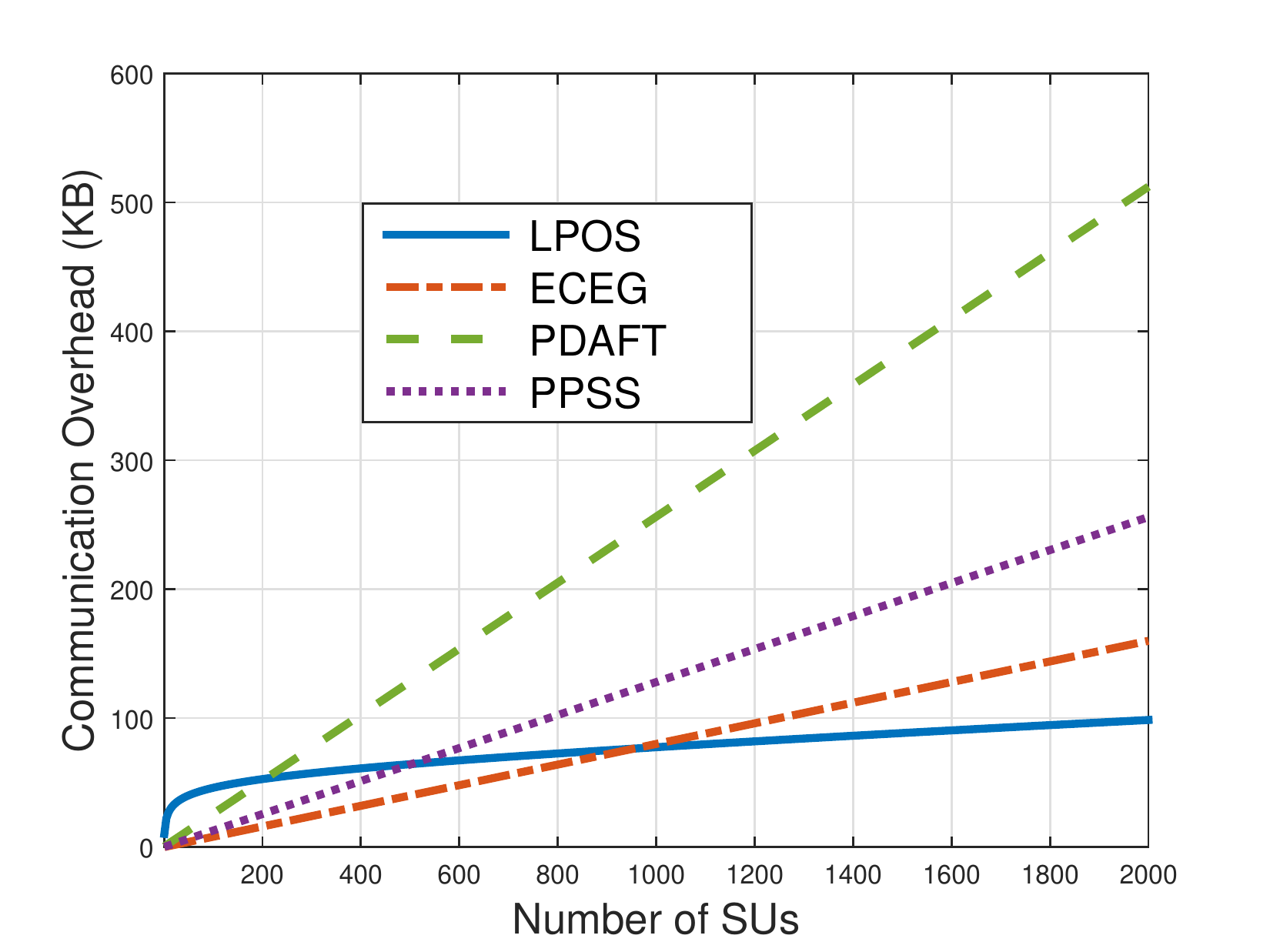}\label{fig:commOv}}
  \subfigure[Properties of different schemes as in Table~\ref{tab:Table1}]{\includegraphics[scale=0.215]{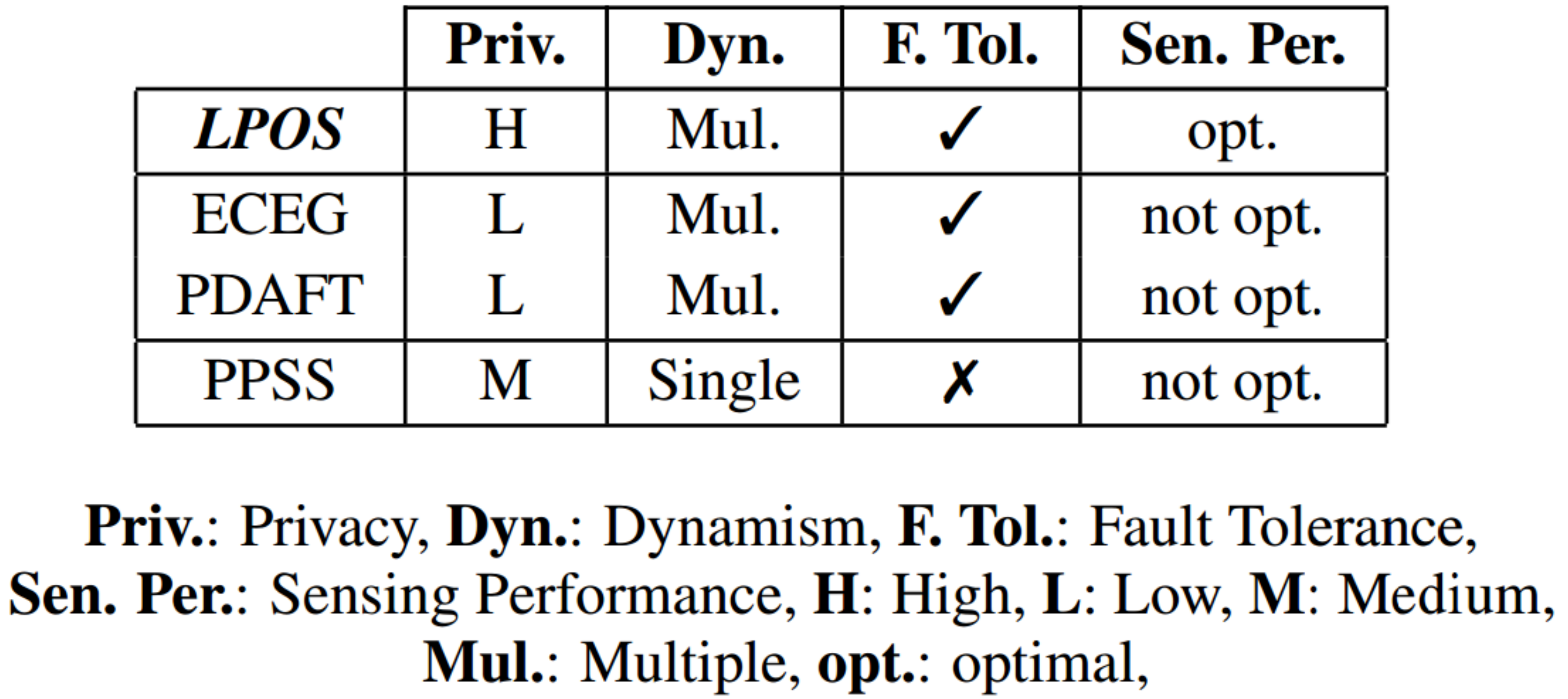}\label{fig:table}}
  \subfigure[Computational Overhead]{\includegraphics[scale=0.36]{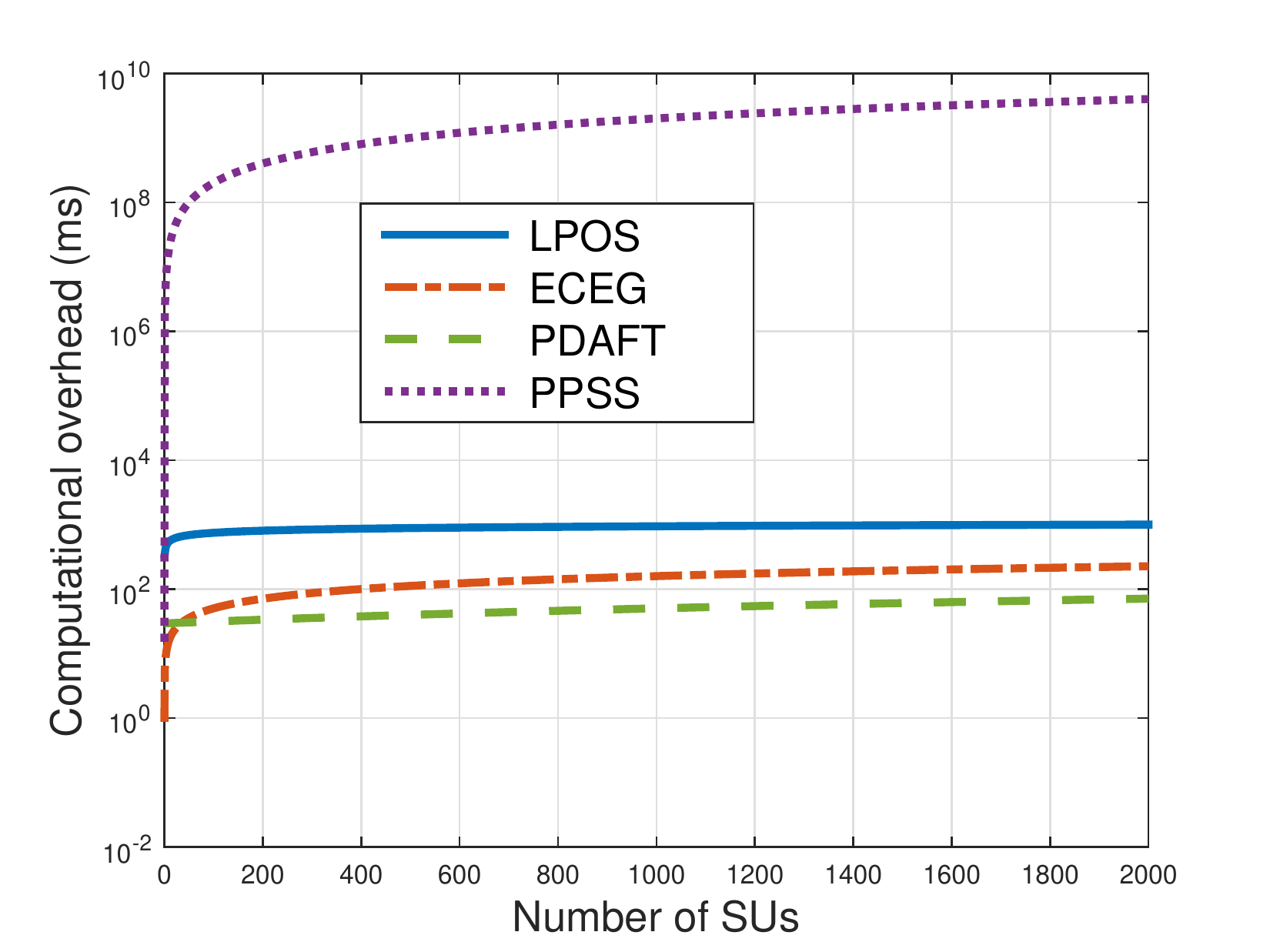}\label{fig:compOv}}
  \caption{Performance Comparison} \label{perfComp}

\end{figure*}






\textbf{Location Privacy, Sensing Accuracy and Reliability}: As shown in Table~\ref{tab:Table1}, \LPOS~achieves the highest level of privacy and decision accuracy among its counterparts. That is, \LPOS~is the only scheme that achieves high location privacy while enabling an optimal spectrum sensing. Moreover, \LPOS~provides fault tolerance and support for dynamism of multiple SUs in the network, which makes it reliable. In addition, \LPOS~also achieves low communication, computation and storage overhead as discussed below.

\textbf{Communication, Computation and Storage Overhead}: Our analytical comparison is summarized in Table \ref{tab:Table3}, which also gives detailed explanations about variables, parameter sizes as well as overhead of building blocks and other schemes included in this comparison. The cost of \LPOS~is determined by \yme, \ope, and {\rm TGECDH}, whose costs are outlined in Table~\ref{tab:Table3}. Notice that the overall cost of \LPOS~is dominated by \yme~protocol, and yet \yme~is invoked only $\mathcal{O}(log(n))$ at the worst case (as explained in Section \ref{sec:ProposedSchemes} in detail). This permits high computational and communication efficiency.

As shown in Fig.~\ref{fig:commOv}\footnote{\scriptsize{Communication overhead of each scheme is calculated by evaluating its corresponding analytical results in Table \ref{tab:Table3} with the parameter sizes given in item (ii)-Table \ref{tab:Table3}.}}, \LPOS~offers the smallest communication overhead among all alternatives for large network sizes thanks to $\mathcal{O}(log(n))$ complexity for all public keys to be transmitted (the small constant $\epsilon_{\ope}$ per user has little impact on the overall communication overhead as seen in the Fig.~\ref{fig:commOv}). It is followed by \ECEG, who has small key sizes for small number of users due to compact ECC parameters. \PPSS~has a high communication overhead, while \PDAFT~incurs extremely large communication overhead due to heavy Pailler encryption.

As shown in Fig.~\ref{fig:compOv}\footnote{\scriptsize{The execution times were measured on a laptop running Ubuntu 14.10 with 8GB of RAM and a core M 1.3 GHz Intel processor, with cryptographic libraries MIRACL \cite{miracl}, Crypto++ \cite{crypto++} and {\em Louismullie}'s Ruby implementation of \ope~\cite{opeRuby}.}}, all compared alternatives are significantly more computationally efficient than \PPSS, while \LPOS~is comparable but little less efficient than \ECEG~and \PDAFT.

Observe that while offering the smallest communication overhead (vital for scalability) and reasonable computation efficiency, \LPOS~is the only scheme that enables optimal spectrum sensing based on voting approach by also providing the highest level of location privacy, fault-tolerance and network dynamism.


\section{Conclusion}
\label{sec:Conclusion}
We design a location privacy preserving scheme for CRNs that achieves high sensing accuracy. Our scheme has several key features, making it more practical, secure, and reliable for large-scale CRNs. When compared to existing approaches, \LPOS~achieves optimal sensing performances with high location privacy while being robust against network dynamism.

\small{
\bibliographystyle{IEEEtran}
\bibliography{IEEEabrv,./references}
}

\end{document}